\documentclass[a4paper,11pt,DIV=12,abstract=true,bibliography=totoc]{scrartcl}
\usepackage[T1]{fontenc}
\usepackage[utf8]{inputenc}
\usepackage[english]{babel}
\usepackage{amssymb,amsfonts,amsmath,amsthm}
\usepackage{mathtools}
\usepackage[hidelinks]{hyperref}
\usepackage{authblk}
\usepackage{booktabs}
\usepackage{subcaption}
\usepackage{lmodern}
\usepackage{enumitem}
\usepackage{bbm}
\usepackage{pgfplots}

\mathtoolsset{showonlyrefs}
\pgfplotsset{compat=1.9}
\newlength\figureheight
\newlength\figurewidth
\usepgfplotslibrary{external}
\newtheorem{theorem}{Theorem}[section]
\newtheorem{lemma}[theorem]{Lemma}
\newtheorem{proposition}[theorem]{Proposition}
\newtheorem{corollary}[theorem]{Corollary}
\theoremstyle{definition}
\newtheorem{definition}[theorem]{Definition}
\newtheorem{example}[theorem]{Example}
\newtheorem{remark}[theorem]{Remark}

\numberwithin{equation}{section}
\numberwithin{table}{section}
\numberwithin{figure}{section}

\newcommand{\R}{\mathbb{R}}

\newcommand{\transp}{\top}
\newcommand{\rmd}{\mathrm{d}}
\newcommand{\rme}{\mathrm{e}}
\newcommand{\ones}{\mathbf{1}_d}
\newcommand{\calA}{\mathcal{A}}
\newcommand{\calF}{\mathcal{F}}
\newcommand{\calG}{\mathcal{G}}
\newcommand{\calK}{\mathcal{K}}
\newcommand{\calN}{\mathcal{N}}
\newcommand{\gam}[2]{Q^{#1}_{#2}}
\newcommand{\muhat}[2]{m^{#1}_{#2}}
\newcommand{\tstrut}{\rule{0pt}{2.6ex}}
\newcommand{\bstrut}{\rule[-0.9ex]{0pt}{0pt}}

\DeclareMathOperator{\E}{\mathbb{E}}

\begin{document}

\title{Robust Utility Maximization in a Multivariate Financial Market with Stochastic Drift}

\author[1]{J\"{o}rn Sass\thanks{\href{mailto:sass@mathematik.uni-kl.de}{sass@mathematik.uni-kl.de}}}
\author[1]{Dorothee Westphal\thanks{\href{mailto:westphal@mathematik.uni-kl.de}{westphal@mathematik.uni-kl.de}}}
\affil[1]{Department of Mathematics, Technische Universit\"{a}t Kaiserslautern}

\date{May~30, 2021}

\maketitle

\begin{abstract}
	We study a utility maximization problem in a financial market with a stochastic drift process, combining a worst-case approach with filtering techniques. Drift processes are difficult to estimate from asset prices, and at the same time optimal strategies in portfolio optimization problems depend crucially on the drift.
	We approach this problem by setting up a worst-case optimization problem with a time-dependent uncertainty set for the drift. Investors assume that the worst possible drift process with values in the uncertainty set will occur. This leads to local optimization problems, and the resulting optimal strategy needs to be updated continuously in time.
	We prove a minimax theorem for the local optimization problems and derive the optimal strategy. Further, we show how an ellipsoidal uncertainty set can be defined based on filtering techniques and demonstrate that investors need to choose a robust strategy to be able to profit from additional information.
	\medskip
	
	
	\smallskip
	
\end{abstract}

\section{Introduction}

Financial market models are usually prone to statistical estimation errors, incomplete information and other reasons for model misspecifications. Especially the drift of asset prices is notoriously difficult to estimate from historical data. Drift processes tend to fluctuate randomly over time, and even for estimating a constant drift with a reasonable degree of precision one needs very long time series, an observation already made by Merton~\cite{merton_1980}. At the same time, trading strategies in portfolio optimization problems depend crucially on the drift. Strategies that are determined based on a misspecified model can therefore perform rather badly in the true financial market setting, see Chopra and Ziemba~\cite{chopra_ziemba_1993} and Kan and Zhou~\cite{kan_zhou_2007}.

There are two main approaches to deal with these problems. On the one hand, it is crucial to approximate the true model as accurately as possible using all the information available. When estimating the hidden drift process the best estimate in a mean-square sense is the conditional mean of the drift given the available information, the so-called \emph{filter}. Observations usually include the stock returns but can also involve external sources of information like news, company reports or ratings. In fact, Merton~\cite{merton_1980} points out that due to the difficulty of estimating expected returns, sources of information other than time series data of market returns are needed to improve estimates. Filtering techniques thus are a way to reduce uncertainty about model parameters.
On the other hand, model uncertainty can be approached by setting up \emph{worst-case optimization} problems. Instead of working with just one particular model, one specifies a range of possible models and tries to optimize the objective, given that for any chosen strategy the worst of all possible models will occur. This leads to robust strategies, i.e.\ strategies that are less vulnerable to the specific choice of the model.

In this paper we combine a worst-case approach with filtering techniques for a utility maximization problem in a financial market with stochastic drift. This is a follow-up paper on Sass and Westphal~\cite{sass_westphal_2020} where a worst-case utility maximization problem for a financial market with constant drift is investigated. In~\cite{sass_westphal_2020} we work with a Black--Scholes market and address an optimization problem of the form
\begin{equation}\label{eq:basic_robust_problem_constant_drift}
	\adjustlimits\sup_{\pi\in\calA_h(x_0)}\inf_{\mu\in K}\E_\mu\bigl[U(X^\pi_T)\bigr],
\end{equation}
where $U\colon\R_+\to\R$ is a utility function, $X^\pi_T$ denotes the terminal wealth achieved when using strategy $\pi$, and $\calA_h(x_0)$ is a class of constrained admissible strategies with initial capital $x_0$. The expectation $\E_\mu[\cdot]$ is with respect to a measure under which the drift of the asset returns is constantly equal to $\mu\in\R^d$, with $d$ denoting the number of risky assets in the market. By $K\subseteq\R^d$ we denote a fixed ellipsoid and speak of the \emph{uncertainty set}. The main result in~\cite{sass_westphal_2020} is a representation of the optimal strategy for~\eqref{eq:basic_robust_problem_constant_drift} in the case of power or logarithmic utility and a corresponding minimax theorem.

In the present paper we generalize the results from Sass and Westphal~\cite{sass_westphal_2020} to a financial market with a stochastic drift process and time-dependent uncertainty sets $K$. This is motivated by the idea that information about the hidden drift process, as e.g.\ obtained from filtering techniques, might change over time. A surplus of information should then be reflected in a smaller uncertainty set.
More precisely, we assume that under the reference measure returns follow the dynamics
\[ \rmd R_t = \nu_t\,\rmd t+\sigma\,\rmd W_t, \]
where the reference drift $(\nu_t)_{t\in[0,T]}$ is adapted to the investor filtration $(\calG_t)_{t\in[0,T]}$ representing the investor's information. This is justified by a separation principle where one performs a filtering step before solving the optimization problem, i.e.\ $(\nu_t)_{t\in[0,T]}$ represents the investor's filter for the drift process. We introduce a time-dependent uncertainty set $(K_t)_{t\in[0,T]}$ that is a set-valued stochastic process adapted to $(\calG_t)_{t\in[0,T]}$, meaning that the investor knows the realization of $K_t$ at time $t$. In our case, $K_t$ is an ellipsoid in $\R^d$.

It is not obvious how to set up a worst-case optimization problem in this time-dependent setting. The problem lies in the fact that the realization of the uncertainty sets $(K_t)_{t\in[0,T]}$ is not known initially but gets revealed over time. A worst-case drift process $(\mu_t)_{t\in[0,T]}$ is characterized by being the worst one with the property that $\mu_t\in K_t$ for all $t\in[0,T]$. However, optimization with respect to this worst-case drift process is not feasible for an investor since it is not known initially. Instead, it makes sense to consider the following local approach. For any fixed $t\in[0,T]$, the current uncertainty set $K_t$ is known. Given this $K_t$, investors take model uncertainty into account by assuming that in the future the worst possible drift process having values in $K_t$ will be realized, i.e.\ the worst drift process from the class
\[ \calK^{(t)}=\bigl\{\mu^{(t)}=(\mu^{(t)}_s)_{s\in[t,T]} \,\big|\, \mu^{(t)}_s\in K_t \text{ and } \mu^{(t)}_s \text{ is }\calG_t\text{-measurable for each }s\in[t,T]\bigr\}. \]
Investors then solve at each time $t\in[0,T]$, given $X_t^\pi =x>0$, the local optimization problem
\begin{equation}\label{eq:basic_robust_problem_local}
	\adjustlimits \sup_{\pi^{(t)}\in\calA_h(t,x)} \inf_{\mu^{(t)}\in \calK^{(t)}} \E_{\mu^{(t)}}\Bigl[U\bigl(X^{t,x,\pi^{(t)}}_T\bigr)\Bigr],
\end{equation}
leading to an optimal strategy $(\pi^{(t),*}_s)_{s\in[t,T]}$. Here $X^{t,x,\pi^{(t)}}_T$ is the terminal wealth for starting at time $t$ with wealth $x$ and following strategy $\pi^{(t)}$. In our continuous-time setting this decision will be revised as soon as $K_t$ changes, possibly continuously in time. The realized optimal strategy of the investor is then given by $\pi^*_t=\pi^{(t),*}_t$ for all $t\in[0,T]$.

The focus of this paper lies in carrying over the results for the robust utility maximization problem with constant drift from Sass and Westphal~\cite{sass_westphal_2020} to the more general model described above. We determine the optimal strategy for~\eqref{eq:basic_robust_problem_local} and prove a local minimax theorem in analogy to Sass and Westphal~\cite[Thm.~3.12]{sass_westphal_2020}.
The difficulty here lies in proving that a worst-case drift exists such that the optimization of the robust objective function is equivalent to optimization with respect to the worst-case drift. The key result of our work is the duality we derive in Theorem~\ref{thm:minimax_theorem} which ensures existence of such a worst-case drift process. Moreover, we are able to specify the form of this drift process explicitly which is what makes it possible to also compute the optimal strategy. This is not an obvious result since duality as in that theorem can only be guaranteed under strong conditions. Results from the literature, e.g.\ from Quenez~\cite{quenez_2004}, do not carry over directly to our setting since the constraint that we put on the admissible strategies leads to a more complicated structure of attainable terminal wealths.

We then show how the time-dependent uncertainty set $(K_t)_{t\in[0,T]}$ can be defined based on the filter $\muhat{}{t}=\E[\mu_t\,|\,\calG_t]$ for various investor filtrations $(\calG_t)_{t\in[0,T]}$. The construction is motivated by confidence regions.
Finally, we compare the optimal strategies for different investor filtrations $(\calG_t)_{t\in[0,T]}$ and investigate which effect a surplus of information has on their performance. By means of a numerical simulation we demonstrate that investors do need to account for model uncertainty by choosing a robustified strategy $\pi^*$. When investors rely on the respective filter only, adding more information leads to a smaller worst-case expected utility since the naive strategy that relies only on the filter is very vulnerable to model misspecifications. Investors need to robustify their strategy by taking model uncertainty into account to be able to profit from additional information. This effect can also be understood as an overconfidence of experts as studied empirically by Heath and Tversky~\cite{heath_tversky_1991}.

\bigskip

Model uncertainty, also called \emph{Knightian uncertainty} in reference to the seminal book by Knight~\cite{knight_1921}, has been addressed in numerous papers. Gilboa and Schmeidler~\cite{gilboa_schmeidler_1989} and Schmeidler~\cite{schmeidler_1989} formulate rigorous axioms on preference relations that account for risk aversion and uncertainty aversion. A robust utility functional in their sense is a mapping
\[ X\mapsto \inf_{Q\in\mathcal{Q}}\E_Q\bigl[U(X)\bigr], \]
where $U$ is a utility function and $\mathcal{Q}$ a convex set of probability measures.
Chen and Epstein~\cite{chen_epstein_2002} give a continuous-time extension of this multiple-priors utility.
Optimal investment decisions under such preferences are investigated in Quenez~\cite{quenez_2004} and Schied~\cite{schied_2005}, building up on Kramkov and Schachermayer~\cite{kramkov_schachermayer_1999, kramkov_schachermayer_2003}. An extension of those results by means of a duality approach is given in Schied~\cite{schied_2007}.
Pflug et al.~\cite{pflug_pichler_wozabal_2012} study risk minimization under model uncertainty.
Papers addressing drift uncertainty in a financial market are Garlappi et al.~\cite{garlappi_uppal_wang_2007} and Biagini and P\i nar~\cite{biagini_pinar_2017}, among others. The latter also focuses on ellipsoidal uncertainty sets. Uncertainty about both drift and volatility is investigated in a recent paper by Pham et al.~\cite{pham_wei_zhou_2018}.

Filtering techniques play a crucial role in utility maximization problems under partial information. There are essentially two models for the drift process that lead to finite-dimensional filters. In the first one the drift is modelled as an Ornstein--Uhlenbeck process, in the second one as a continuous-time Markov chain. The filters are the well-known Kalman and Wonham filter, respectively, see e.g.\ Elliott et al.~\cite{elliott_aggoun_moore_1995} and Liptser and Shiryaev~\cite{liptser_shiryaev_1974}.

\bigskip

The paper is organized as follows. Since this is a follow-up paper on Sass and Westphal~\cite{sass_westphal_2020} that generalizes results for a financial market with constant drift to one with stochastic drift we recap the main results of \cite{sass_westphal_2020} in Section~\ref{sec:recap_of_results_for_constant_drift}. For the convenience of the reader the later sections then refer to Section~2. In Section~\ref{sec:generalized_financial_market_model} we set up the generalized financial market model with stochastic drift process and state our local worst-case optimization problem. Section~\ref{sec:solution_to_the_robust_utility_maximization_problem} solves this problem in several steps. We provide representations of the optimal strategy that will be realized by an investor whose information about the drift process changes continuously in time and of the worst-case drift process. Further, we prove a minimax theorem for the local optimization problems. In Section~\ref{sec:construction_of_uncertainty_sets_via_filters} we explain how filtering techniques can be used to set up time-dependent uncertainty sets, motivated by confidence regions. We also compare the performance of the optimal strategies for different investor filtrations by means of a numerical simulation.

\section{Recap of Results for Constant Drift}\label{sec:recap_of_results_for_constant_drift}

This paper builds up on Sass and Westphal~\cite{sass_westphal_2020} and generalizes results for a financial market with constant drift to a model with a stochastic drift process. For the convenience of the reader we recap the main results of Sass and Westphal~\cite{sass_westphal_2020} in this section. This is in order to show in the following sections what carries over and where we have to provide new arguments. By referring to the cited results in this section, this can be conveniently done in a self-contained way.

\subsection{Financial market model}

The paper~\cite{sass_westphal_2020} deals with a continuous-time financial market with one risk-free and various risky assets. Let $T>0$ denote some finite investment horizon and let $(\Omega, \calF, \mathbb{F}, \mathbb{P})$ be a filtered probability space where the filtration $\mathbb{F}=(\calF_t)_{t\in[0,T]}$ satisfies the usual conditions. All processes are assumed to be $\mathbb{F}$-adapted.
The risk-free asset $S^0$ is of the form $S^0_t=\rme^{rt}$, $t\in[0,T]$, where $r\in\R$ is the deterministic risk-free interest rate.
Aside from the risk-free asset, investors can also invest in $d\geq 2$ risky assets. Their prices are given by constant initial prices $S_0^i >0$ and evolve according to $\rmd S_t^i = S_t^i \,\rmd R_t^i$, $i=1, \ldots, d$, where the return process $R=(R^1,\dots,R^d)^\transp$ is defined by
\[ \rmd R_t = \nu\,\rmd t + \sigma\,\rmd W_t, \quad R_0=0, \]
where $W=(W_t)_{t\in[0,T]}$ is an $m$-dimensional Brownian motion under $\mathbb{P}$ with $m\geq d$. Further, $\nu\in\R^d$ and $\sigma\in\R^{d\times m}$, where it is assumed that $\sigma$ has full rank equal to $d$.

Model uncertainty is introduced by assuming that the true drift of the stocks is only known to be an element of some set $K\subseteq\R^d$ with $\nu\in K$ and that investors want to maximize their worst-case expected utility when the drift takes values within $K$. The value $\nu$ can be thought of as an estimate for the drift that was for instance obtained from historical stock prices. Changing the drift from $\nu$ to some $\mu\in K$ can be expressed by a change of measure. For this purpose, let the process $(Z^\mu_t)_{t\in[0,T]}$ be defined by
\[ Z^\mu_t = \exp\Bigl(\theta(\mu)^\transp W_t -\frac{1}{2}\lVert\theta(\mu)\rVert^2 t\Bigr), \]
where $\theta(\mu)=\sigma^\transp(\sigma\sigma^\transp)^{-1}(\mu-\nu)$. One can then define a new measure $\mathbb{P}^\mu$ by setting $\frac{\rmd \mathbb{P}^\mu}{\rmd \mathbb{P}} = Z^\mu_T$. Note that since $\theta(\mu)$ is a constant, the process $(Z^\mu_t)_{t\in[0,T]}$ is a strictly positive martingale. Therefore, $\mathbb{P}^\mu$ is a probability measure that is equivalent to $\mathbb{P}$ and it follows from Girsanov's Theorem that the process $(W^\mu_t)_{t\in[0,T]}$, defined by $W^\mu_t = W_t-\theta(\mu)t$, is a Brownian motion under $\mathbb{P}^\mu$. The return dynamics can therefore be rewritten as
\begin{equation*}
	\rmd R_t = \nu\,\rmd t + \sigma\,\rmd W_t = \nu\,\rmd t + \sigma\bigl(\rmd W^\mu_t+\theta(\mu)\,\rmd t\bigr) = \mu\,\rmd t + \sigma\,\rmd W^\mu_t,
\end{equation*}
hence a change of measure from $\mathbb{P}$ to $\mathbb{P}^\mu$ corresponds to changing the drift in the return dynamics from $\nu$ to $\mu$. In the following, let $\E_\mu[\cdot]$ denote the expectation under measure $\mathbb{P}^\mu$ and $\E[\cdot]=\E_\nu[\cdot]$ the expectation under the reference measure $\mathbb{P}=\mathbb{P}^\nu$.

An investor's trading decisions are described by a self-financing trading strategy $(\pi_t)_{t\in[0,T]}$ with values in $\R^d$. The entry $\pi^i_t$, $i=1, \dots, d$, is the proportion of wealth invested in asset $i$ at time $t$. The corresponding wealth process $(X^\pi_t)_{t\in[0,T]}$ given initial wealth $x_0>0$ can then be described by the stochastic differential equation
\[ \rmd X^\pi_t = X^\pi_t\Bigl( r\,\rmd t + \pi_t^\transp(\mu-r\ones)\,\rmd t + \pi_t^\transp \sigma\,\rmd W^\mu_t \Bigr), \quad X^\pi_0=x_0, \]
for any $\mu\in K$, where $\ones$ denotes the $d$-dimensional vector with all entries equal to~$1$.
Trading strategies are required to be $\mathbb{F}^R$-adapted, where $\mathbb{F}^R=(\calF^R_t)_{t\in[0,T]}$ for $\calF^R_t=\sigma((R_s)_{s\in[0,t]})$. The basic admissibility set is defined as
\[ \calA(x_0) = \biggl\{(\pi_t)_{t\in[0,T]} \;\bigg|\; \pi \text{ is } \mathbb{F}^R\text{-adapted}, \; X^\pi_0=x_0, \; \E_\mu\biggl[\int_0^T \lVert\sigma^\transp\pi_t\rVert^2\,\rmd t\biggr]<\infty \text{ for all } \mu\in K\biggr\}. \]
The paper Sass and Westphal~\cite{sass_westphal_2020} considers investors with power or logarithmic utility, using the notation $U_\gamma\colon\R_+\to\R$ for $\gamma\in(-\infty,1)$, where $U_\gamma(x)=\frac{x^\gamma}{\gamma}$ for $\gamma\neq 0$ denotes power utility and $U_0(x)=\log(x)$ is the logarithmic utility function.
Investors with a robust approach to the portfolio optimization problem would try to maximize
\[ \inf_{\mu\in K} \E_\mu\bigl[U_\gamma(X^\pi_T)\bigr] \]
among the admissible strategies. It is quite straightforward to show that as soon as $r\ones\in K$, the strategy $(\pi_t)_{t\in[0,T]}$ with $\pi_t=0$ for all $t\in[0,T]$ is optimal in the class of admissible strategies $\calA(x_0)$. This observation is proven in Sass and Westphal~\cite[Prop.~2.1]{sass_westphal_2020} and implies that as the level of uncertainty exceeds a certain threshold, it will be optimal for investors to not invest anything in the stocks and everything in the risk-free asset. For finding less conservative strategies that still take into account model uncertainty a constraint on the admissible strategies is introduced that prevents a pure bond investment. Consider for some $h>0$ the admissibility set
\[ \calA_h(x_0)=\bigl\{ \pi\in\calA(x_0) \,\big|\, \langle\pi_t,\ones\rangle = h \text{ for all } t\in[0,T] \bigr\}. \]
Taking $h=1$ would imply that investors are not allowed to invest anything in the risk-free asset. They must then distribute all of their wealth among the risky assets.
Sass and Westphal~\cite{sass_westphal_2020} study the case where the uncertainty set is an ellipsoid in $\R^d$ centered around the reference parameter $\nu$, i.e.\
\begin{equation}\label{eq:uncertainty_ellipsoid}
	K=\bigl\{ \mu\in\R^d \,\big|\, (\mu-\nu)^\transp \Gamma^{-1}(\mu-\nu) \leq \kappa^2 \bigr\}.
\end{equation}
Here, $\kappa>0$, $\nu\in\R^d$, and $\Gamma\in\R^{d\times d}$ is symmetric and positive definite. For $\Gamma=I_d$ one simply gets a ball in the Euclidean norm with radius $\kappa$ and center $\nu$. Another special case discussed in the literature is $\Gamma=\sigma\sigma^\transp$, see e.g.\ Biagini and P\i nar~\cite{biagini_pinar_2017}.
The value of $\kappa$ determines the size of the ellipsoid. Higher values of $\kappa$ correspond to more uncertainty about the true drift. The robust utility maximization problem over the constrained strategies $\pi\in\calA_h(x_0)$ can then be written in the form
\begin{equation}\label{eq:robust_problem}
	\adjustlimits \sup_{\pi\in\calA_h(x_0)} \inf_{\mu\in K} \E_\mu\bigl[U_\gamma(X^\pi_T)\bigr].
\end{equation}

\subsection{Solution of the non-robust problem}

To solve the optimization problem in~\eqref{eq:robust_problem} Sass and Westphal~\cite{sass_westphal_2020} first address the non-robust constrained utility maximization problem under a fixed parameter $\mu\in\R^d$. For better readability the following notation is introduced.

\begin{definition}\label{def:matrix_A_vector_c}
	Denote with $D$ the matrix
	\[ D =
	\begin{pmatrix}
		1	&	& 0	& -1 \\
			&\ddots	&	&\vdots \\
		0	&	& 1	& -1
	\end{pmatrix}\in\R^{(d-1)\times d} \]
	and define the matrix $A\in\R^{d\times d}$ and the vector $c\in\R^d$ by
	\begin{align*}
		A &= D^\transp(D\sigma\sigma^\transp D^\transp)^{-1}D, \\
		c &= e_d-D^\transp(D\sigma\sigma^\transp D^\transp)^{-1}D\sigma\sigma^\transp e_d = (I_d-A\sigma\sigma^\transp)e_d.
	\end{align*}
\end{definition}

The following result gives the optimal strategy for the non-robust problem and can be found in Sass and Westphal~\cite[Prop.~3.4]{sass_westphal_2020}.

\begin{proposition}\label{prop:optimal_strategy_non-robust}
	Let $\mu\in\R^d$. Then the optimal strategy for the optimization problem
	\[ \sup_{\pi\in\calA_h(x_0)} \E_\mu\bigl[U_\gamma(X^\pi_T)\bigr] \]
	is the strategy $(\pi_t)_{t\in[0,T]}$ with
	\[ \pi_t = \frac{1}{1-\gamma}A\mu +hc \]
	for all $t\in[0,T]$, with $A$ and $c$ as in Definition~\ref{def:matrix_A_vector_c}.
\end{proposition}
      
In the proof the $d$-dimensional constrained problem is reduced to a $(d-1)$-dimensional unconstrained problem. Using the form of the optimal strategy in the $(d-1)$-dimensional market which is known from Merton~\cite{merton_1969} yields the following representation for the optimal expected utility from terminal wealth. This result is given in Sass and Westphal~\cite[Cor.~3.5]{sass_westphal_2020}.

\begin{corollary}\label{cor:optimal_utility_non-robust}
	Let $\mu\in\R^d$. Then the optimal expected utility from terminal wealth is
	\begin{equation*}
		\begin{aligned}
			\sup_{\pi\in\calA_h(x_0)} &\E_\mu\bigl[U_\gamma(X^\pi_T)\bigr] \\
			&=
			\begin{dcases}
				\frac{x_0^\gamma}{\gamma}\exp\Bigl(\gamma T\Bigl( \widetilde{r}+\frac{1}{2(1-\gamma)}\bigl(\widetilde{\mu}-\widetilde{r}\mathbf{1}_{d-1}\bigr)^\transp(\widetilde{\sigma}\widetilde{\sigma}^\transp )^{-1}\bigl(\widetilde{\mu}-\widetilde{r}\mathbf{1}_{d-1}\bigr)\Bigr)\Bigr), &\gamma\neq 0,\\
				\log(x_0) + \Bigl( \widetilde{r}+\frac{1}{2}\bigl(\widetilde{\mu}-\widetilde{r}\mathbf{1}_{d-1}\bigr)^\transp(\widetilde{\sigma}\widetilde{\sigma}^\transp )^{-1}\bigl(\widetilde{\mu}-\widetilde{r}\mathbf{1}_{d-1}\bigr) \Bigr)T, &\gamma=0,
			\end{dcases}
		\end{aligned}
	\end{equation*}
	where
	\begin{equation}\label{eq:recall_substitution_r_mu_sigma}
		\begin{aligned}
			\widetilde{\sigma}&=D\sigma, \\
			\widetilde{r}&=(1-h)r+he_d^\transp\mu-\frac{1}{2}(1-\gamma)\lVert h\sigma^\transp e_d \rVert^2, \\
			\widetilde{\mu}&=D\mu - h(1-\gamma)D\sigma\sigma^\transp e_d+\widetilde{r}\mathbf{1}_{d-1}.
		\end{aligned}
	\end{equation}
\end{corollary}

\subsection{The worst-case parameter}

In a next step one may ask what the worst possible parameter $\mu$ would be for the investor, given that she reacts optimally, i.e.\ by applying the strategy from Proposition~\ref{prop:optimal_strategy_non-robust}. This corresponds to solving the dual problem
\[ \adjustlimits \inf_{\mu\in K} \sup_{\pi\in\calA_h(x_0)} \E_\mu\bigl[U_\gamma(X^\pi_T)\bigr]. \]
Note that at this point it is not clear whether equality holds between the original problem and the corresponding dual problem. The following result for the solution of the dual problem is given in Sass and Westphal~\cite[Thm.~3.8]{sass_westphal_2020}. Let us decompose $\Gamma=\tau\tau^\transp$ for a nonsingular matrix $\tau\in\R^{d\times d}$.

\begin{theorem}\label{thm:solution_of_the_inf_sup_problem}
	Let $0=\lambda_1<\lambda_2\leq\cdots\leq\lambda_d$ denote the eigenvalues of $\tau^\transp A\tau$, and let
	\[ v_1=\frac{1}{\lVert \tau^{-1}\ones \rVert}\tau^{-1}\ones, v_2,\dots,v_d\in\R^d \]
	denote the respective orthogonal eigenvectors with $\lVert v_i\rVert=1$ for all $i=1,\dots, d$.
	Then
	\[ \adjustlimits \inf_{\mu\in K} \sup_{\pi\in\calA_h(x_0)} \E_\mu\bigl[U_\gamma(X^\pi_T)\bigr] = \E_{\mu^*}\bigl[U_\gamma(X^{\pi^*}_T)\bigr], \]
	where
	\[ \mu^*=\nu-\tau\sum_{i=1}^d \biggl( \frac{\lambda_i}{1-\gamma}+\frac{h}{\psi(\kappa)\lVert \tau^{-1}\ones \rVert} \biggr)^{-1}\biggl\langle h\tau^\transp c+\frac{\lambda_i}{1-\gamma}\tau^{-1}\nu, v_i\biggr\rangle v_i \]
	for $\psi(\kappa)\in(0,\kappa]$ that is uniquely determined by $\lVert\tau^{-1}(\mu^*-\nu)\rVert=\kappa$, and where $(\pi^*_t)_{t\in[0,T]}$ is defined by
	\[ \pi^*_t = \frac{1}{1-\gamma}A\mu^* +hc \]
	for all $t\in[0,T]$.
\end{theorem}

The preceding theorem solves the problem
\begin{equation}\label{eq:the_inf_sup_problem}
	\adjustlimits \inf_{\mu\in K} \sup_{\pi\in\calA_h(x_0)} \E_\mu\bigl[U_\gamma(X^\pi_T)\bigr].
\end{equation}
This is the corresponding dual problem to the original optimization problem
\begin{equation}\label{eq:the_sup_inf_problem}
	\adjustlimits \sup_{\pi\in\calA_h(x_0)} \inf_{\mu\in K} \E_\mu\bigl[U_\gamma(X^\pi_T)\bigr],
\end{equation}
but in general the values of these two problems do not coincide. There are, of course, special cases in which the supremum and the infimum do interchange. Those results are called \emph{minimax theorems} in the literature. In a portfolio optimization setting that is similar to ours a minimax theorem has been shown in Quenez~\cite{quenez_2004}, building up on the theory by Kramkov and Schachermayer~\cite{kramkov_schachermayer_1999}. Due to the constraint $\langle\pi_t,\ones\rangle=h$ for all $t\in[0,T]$, the result from Quenez~\cite{quenez_2004} does not apply directly to our setting. It is possible, however, to use the knowledge about the optimal strategy for~\eqref{eq:the_inf_sup_problem} to show that it indeed also solves~\eqref{eq:the_sup_inf_problem} and that in this case, the supremum and the infimum can be interchanged.

\subsection{A minimax theorem}

The following representation of $\pi^*$, given in Sass and Westphal~\cite[Lem.~3.10]{sass_westphal_2020}, is useful for proving a minimax theorem.

\begin{lemma}\label{lem:representation_of_pi_star}
	The strategy $\pi^*$ from Theorem~\ref{thm:solution_of_the_inf_sup_problem} satisfies
	\[ \pi^*_t = -\frac{h}{\psi(\kappa)\lVert \tau^{-1}\ones \rVert}\Gamma^{-1}(\mu^*-\nu) \]
	for all $t\in[0,T]$.
\end{lemma}

The preceding lemma characterizes the strategy $\pi^*$ that is optimal for the parameter $\mu^*$. Vice versa, $\mu^*$ is also the worst possible drift parameter, given that an investor applies strategy $\pi^*$. This is shown in Sass and Westphal~\cite[Prop.~3.11]{sass_westphal_2020}.

\begin{proposition}\label{prop:mu_star_is_worst_for_pi_star}
	The parameter $\mu$ that attains the minimum in
	\[ \inf_{\mu\in K} \E_\mu\bigl[U_\gamma(X^{\pi^*}_T)\bigr] \]
	is $\mu^*$, i.e.\ $\mu^*$ is the worst possible parameter, given that an investor chooses strategy $\pi^*$.
\end{proposition}

It then follows that the point $(\pi^*,\mu^*)$ is a \emph{saddle point} of the problem, i.e.\ it holds
\[ \E_{\mu^*}\bigl[U_\gamma(X^{\pi}_T)\bigr] \leq \E_{\mu^*}\bigl[U_\gamma(X^{\pi^*}_T)\bigr] \leq \E_{\mu}\bigl[U_\gamma(X^{\pi^*}_T)\bigr] \]
for all $\mu\in K$ and $\pi\in\calA_h(x_0)$. This property is essential for proving the following minimax theorem, given in Sass and Westphal~\cite[Thm.~3.12]{sass_westphal_2020}.

\begin{theorem}\label{thm:duality_result}
	Let $K=\{ \mu\in\R^d \,|\, (\mu-\nu)^\transp \Gamma^{-1}(\mu-\nu) \leq \kappa^2 \}$. Then
	\[ \adjustlimits \sup_{\pi\in\calA_h(x_0)} \inf_{\mu\in K} \E_\mu\bigl[U_\gamma(X^\pi_T)\bigr] = \E_{\mu^*}\bigl[U_\gamma(X^{\pi^*}_T)\bigr] = \adjustlimits \inf_{\mu\in K} \sup_{\pi\in\calA_h(x_0)} \E_\mu\bigl[U_\gamma(X^\pi_T)\bigr], \]
	where $\mu^*$ and $\pi^*$ are defined as in Theorem~\ref{thm:solution_of_the_inf_sup_problem}.
\end{theorem}

\section{Generalized Financial Market Model}\label{sec:generalized_financial_market_model}

In the following we generalize the approach from Sass and Westphal~\cite{sass_westphal_2020} to a financial market model where the drift is a stochastic process instead of a constant. To account for a change in information about the drift we also introduce time-dependence in the uncertainty set.
The basic idea is that the available information in the market, for instance the observed asset returns or external sources of information, are used to estimate the true drift based on filtering techniques and to set up a corresponding uncertainty set $K_t$ at any time $t\in[0,T]$. Given $K_t$, investors then take model uncertainty into account by assuming that in the future the worst possible drift process $(\mu^{(t)}_s)_{s\in[t,T]}$ with values in $K_t$ will be realized. In our continuous-time setting the decision about the uncertainty set will be revised as soon as the information about the true drift changes, so in the extreme case continuously in time.

\subsection{Reference model}

Before stating our generalized financial market, we make an observation that justifies the setup of the model.

\begin{remark}
	Suppose that the ``true'' dynamics of the $d$-dimensional return process $R$ are given by
	\[ \rmd R_t=\mu_t\,\rmd t + \sigma\,\rmd W_t, \quad R_0=0, \]
	for some stochastic drift process $(\mu_t)_{t\in[0,T]}$, an $m$-dimensional Brownian motion $(W_t)_{t\in[0,T]}$, $m\geq d$, and some $\sigma\in\R^{d\times m}$ with full rank. Assume further that the information of an investor is given by the investor filtration $\mathbb{G}=(\calG_t)_{t\in[0,T]}$. The investor's best estimator for $\mu$ is then the conditional mean $m_t:=\E[\mu_t\,|\,\calG_t]$ and one can rewrite the dynamics of the return process as
	\[ \rmd R_t=m_t\,\rmd t + \sigma\,\rmd V_t, \]
	where the so-called innovations process $(V_t)_{t\in[0,T]}$ is a $\mathbb{G}$-adapted Brownian motion.
	For instance, in the setting where an investor observes only the return process $R$, the process $(m_t)_{t\in[0,T]}$ would be the Kalman filter.

	In the following, we set up our continuous-time financial market model working directly with the innovations process and therefore assuming a $\mathbb{G}$-adapted drift process.
	The separation principle that we use here by filtering first and then performing the optimization is a common approach for dealing with partial information.
\end{remark}

We fix an investment horizon $T>0$ and some filtered probability space $(\Omega, \calF, \mathbb{F}, \mathbb{P})$ where the filtration $\mathbb{F}=(\calF_t)_{t\in[0,T]}$ satisfies the usual conditions. All processes are assumed to be $\mathbb{F}$-adapted. We assume that an investor's information is described by the investor filtration $\mathbb{G}=(\calG_t)_{t\in[0,T]}$ with $\calG_t\subseteq\calF_t$ for all $t\in[0,T]$.
We consider, as before, a financial market with one risk-free and $d\geq 2$ risky assets. The risk-free asset $S^0$ evolves as
\[ \rmd S^0_t = S^0_t r\,\rmd t, \quad S^0_0 = 1, \]
where $r>0$ is the deterministic risk-free interest rate. The prices of the risky assets  evolve according to $\rmd S_t^i = S_t^i \,\rmd R_t^i$ for $S_0^i>0$, $i=1, \ldots, d$, where the returns $R^1,\dots,R^d$ of the risky assets follow the dynamics
\begin{equation}\label{eq:dynamics_of_S_time_dependent_drift}
	\rmd R_t = \nu_t\,\rmd t + \sigma\,\rmd W_t, \quad R_0=0,
\end{equation}
where $R=(R^1,\dots,R^d)^\transp$. Here, $(W_t)_{t\in[0,T]}$ is an $m$-dimensional Brownian motion under $\mathbb{P}$, $m\geq d$. Note that the volatility matrix $\sigma\in\R^{d\times m}$ in~\eqref{eq:dynamics_of_S_time_dependent_drift} is constant. Further, we assume that $\sigma$ has full rank equal to $d$. In contrast to the volatility, the drift might change in the course of time. We assume that $(\nu_t)_{t\in[0,T]}$ is an $\R^d$-valued $\mathbb{G}$-adapted stochastic process and think of $(\nu_t)_{t\in[0,T]}$ as an estimation for the true drift process given all available information. We speak of $(\nu_t)_{t\in[0,T]}$ as the \emph{reference drift}.

\begin{remark}
	We could generalize this model by replacing the constant volatility $\sigma$ by a stochastic process $(\sigma_t)_{t\in[0,T]}$ that is \emph{observable} by the investor, i.e.\ $\mathbb{G}$-adapted. Our techniques and the results in the following section carry over to this setting. It would just be necessary to put assumptions on the volatility process to ensure that the change of measure in Section~\ref{subs:uncertainty_sets_and_change_of_measure} below is well-defined.
	
	In a model where the stochastic volatility $(\sigma_t)_{t\in[0,T]}$ is \emph{not} observable by the investor, our techniques do not carry over in general. The reason is that we need models where we can write down the (non-robust) optimal strategy of an investor explicitly. Portfolio optimization under stochastic volatility is much more complicated since one has to deal with truly incomplete markets. Pham and Quenez~\cite{pham_quenez_2001} state assumptions on the model needed to obtain solutions. Similar restrictions might work for our case. These also have to take into account observability issues which essentially require that the volatility process can be observed one-to-one from the quadratic variation of the stock returns.
	
	However, for portfolio optimization the drift is the crucial factor and, at the same time, the one that is difficult to estimate from historical asset price data. Drift processes tend to fluctuate randomly over time and even if they were constant, long time series would be needed to estimate this parameter with a satisfactory degree of precision. On the contrary, volatility can be estimated reasonably well from return observations. For these reasons, our model with constant volatility serves as a suitable first approximation for more general models that allows an exact study of the influence of model uncertainty on the optimal strategy.
\end{remark}

\subsection{Uncertainty sets and change of measure}\label{subs:uncertainty_sets_and_change_of_measure}

As before, we are concerned with investors who are uncertain about the true drift. They are aware that $(\nu_t)_{t\in[0,T]}$ in~\eqref{eq:dynamics_of_S_time_dependent_drift} might not be the true drift process. In utility maximization problems they want to maximize their worst-case expected utility, given that the true drift process is in a way ``close'' to $\nu$.
To model the uncertainty about the drift we specify the ellipsoidal sets
\[ K_t = \bigl\{ \mu\in\R^d \,\big|\, (\mu-\nu_t)^\transp\Gamma_t^{-1}(\mu-\nu_t)\leq \kappa_t^2 \bigr\}, \quad t\in[0,T], \]
where $(\Gamma_t)_{t\in[0,T]}$ is a $\mathbb{G}$-adapted stochastic process of symmetric and positive-definite matrices $\Gamma_t\in\R^{d\times d}$ and $(\kappa_t)_{t\in[0,T]}$ is $\mathbb{G}$-adapted with $\kappa_t>0$ for each $t\in[0,T]$.
The set $K_t$ is determined at time $t\in[0,T]$ by taking the available information about the true drift process into account, for example based on filtering techniques. The process $(K_t)_{t\in[0,T]}$ is a $\mathbb{G}$-adapted set-valued process, therefore the investor knows the realization of $K_t$ at time $t\in[0,T]$.

Given this $K_t$, investors then take model uncertainty into account by assuming that in the future the worst possible drift process having values in $K_t$ will be realized. We denote this worst-case future drift by $(\mu^{(t),*}_s)_{s\in[t,T]}$. This allows for some deterministic dynamics given $K_t$, i.e.\ the $\mu^{(t),*}_s$ for any $s\in[t,T]$ are $\calG_t$-measurable. The worst-case optimization problem then leads to an optimal strategy $(\pi^{(t),*}_s)_{s\in[t,T]}$, determined at time $t$. In our continuous-time setting this decision will be revised as soon as $K_t$ changes, possibly continuously in time. The realized worst-case drift process $(\mu^*_t)_{t\in[0,T]}$ and optimal strategy $(\pi^*_t)_{t\in[0,T]}$ are then given by
\[ \mu^*_t=\mu^{(t),*}_t, \quad \pi^*_t=\pi^{(t),*}_t \]
for any $t\in[0,T]$. If $\mu^*$ and $\pi^*$ are uniquely determined, then they are by construction $\mathbb{G}$-adapted.
This is not so much a game setting but rather a way how the investor determines the worst case. It is a mixture of using estimation methods and taking model uncertainty into account.

The optimization problem can be derived only locally for each $t\in[0,T]$. In detail, the setup looks as follows. At time $t\in[0,T]$ investors assume that the future drift process will be the worst one within the class
\[ \calK^{(t)}=\bigl\{\mu^{(t)}=(\mu^{(t)}_s)_{s\in[t,T]} \,\big|\, \mu^{(t)}_s\in K_t \text{ and } \mu^{(t)}_s \text{ is }\calG_t\text{-measurable for each }s\in[t,T]\bigr\}. \]
For each $\mu=\mu^{(t)}\in\calK^{(t)}$ we can construct a new measure by defining the $\R^m$-valued process $(\theta_s(\mu))_{s\in[0,T]}$ with
\[ \theta_s(\mu)=
\begin{cases}
	0, &s<t,\\
	\sigma^\transp(\sigma\sigma^\transp)^{-1}(\mu_s-\nu_s), &s\geq t,
\end{cases} \]
and
\[ Z^\mu_s = \exp\biggl(\int_0^s\theta_u(\mu)^\transp\,\rmd W_u -\frac{1}{2}\int_0^s\lVert\theta_u(\mu)\rVert^2\,\rmd u\biggr) \]
for $s\in[0,T]$. We then define the new probability measure $\mathbb{P}^\mu$ by
\[ \frac{\rmd \mathbb{P}^{\mu}}{\rmd \mathbb{P}} = Z^\mu_T \]
and note that, under $\mathbb{P}^\mu$, the process $(W^\mu_s)_{s\in[0,T]}$ with
\[ W^\mu_s = W_s-\int_0^s \theta_u(\mu)\,\rmd u \]
for $s\in[0,T]$ is a Brownian motion by Girsanov's Theorem. Note that due to boundedness of $K_t$ the process $\theta(\mu)$ is bounded and therefore $(Z^\mu_s)_{s\in[0,T]}$ is a true martingale.
The change of measure causes a change in the drift on the interval $[t,T]$ only. For our optimization problems this is the only relevant time interval since we condition on $\calG_t$.
For $s\in[t,T]$ we can rewrite the dynamics of the asset returns as
\[ \rmd R_s = \nu_s\,\rmd s + \sigma\,\rmd W_s = \mu_s\,\rmd s + \sigma\,\rmd W^\mu_s, \]
which means that under $\mathbb{P}^\mu$ the future drift of the stocks is given by $(\mu_s)_{s\in[t,T]}$. We write $\E_\mu[\cdot]=\E_{\mu^{(t)}}[\cdot]$ for expectation under the measure $\mathbb{P}^\mu$.

\subsection{Local optimization problem}

An investor's behavior in the time interval $[t,T]$ is described by a self-financing trading strategy $\pi^{(t)}=(\pi^{(t)}_s)_{s\in[t,T]}$. The class of admissible trading strategies, given that the investor has wealth $x>0$ at time $t$, is
\begin{equation*}
	\begin{aligned}
		\calA(t,x) = \biggl\{\pi^{(t)}=(\pi^{(t)}_s)_{s\in[t,T]} \;\bigg|&\; \pi^{(t)} \text{ is } \mathbb{G}\text{-adapted}, \; X^\pi_t=x,\\
		&\E_{\mu^{(t)}}\biggl[\int_t^T\! \lVert\sigma^\transp\pi^{(t)}_s\rVert^2\,\rmd s\biggr]<\infty \text{ for all } \mu^{(t)}\in \calK^{(t)}\biggr\}.
	\end{aligned}
\end{equation*}
We will restrict these strategies by imposing, as in Sass and Westphal~\cite{sass_westphal_2020}, a constraint that prevents a pure bond investment. For any $h>0$ we define the set
\[ \calA_h(t,x)=\bigl\{ \pi^{(t)}\in\calA(t,x) \,\big|\, \langle\pi^{(t)}_s,\ones\rangle=h \text{ for all }s\in[t,T]\bigr\}. \]
For an investor choosing strategy $\pi=\pi^{(t)}\in\calA(t,X^\pi_t)$ the terminal wealth can be written as
\[ X^\pi_T = X^\pi_t\exp\biggl(\int_t^T \Bigl(r+\pi_s^\transp(\mu_s-r\ones)-\frac{1}{2}\lVert\sigma^\transp\pi_s\rVert^2\Bigr)\rmd s + \int_t^T \pi_s^\transp\sigma\,\rmd W^\mu_s\biggr). \]

We are now able to state our utility maximization problem. In the following we write $X^{t,x,\pi^{(t)}}_s$ for the wealth at $s\in[t,T]$ when starting at $t$ with $x$ and using strategy $\pi^{(t)}$. At time $t$ the local optimization problem for $x>0$ then reads
\begin{equation}\label{eq:value_function_robust_power_constrained_time-dependent}
	\adjustlimits \sup_{\pi^{(t)}\in\calA_h(t,x)} \inf_{\mu^{(t)}\in \calK^{(t)}} \E_{\mu^{(t)}}\Bigl[U_\gamma\bigl(X^{t,x,\pi^{(t)}}_T\bigr)\Bigr].
\end{equation}
Here, $U_\gamma$ with $\gamma\in(-\infty,1)$ again denotes the power utility function $U_\gamma(x)=\frac{x^\gamma}{\gamma}$ if $\gamma\neq 0$, and logarithmic utility $U_0(x)=\log(x)$ if $\gamma=0$.

\begin{remark}
	In the case where $K_t= \{ \mu\in\R^d \,|\, (\mu-\nu)^\transp\Gamma^{-1}(\mu-\nu)\leq \kappa^2 \}$ for all $t\in[0,T]$, i.e.\ where our reference drift is simply a constant $\nu$, and also the matrix $\Gamma_t=\Gamma$ as well as the radius $\kappa_t=\kappa$ are constant in time, we obtain the setting from Section~\ref{sec:recap_of_results_for_constant_drift} as a special case.
\end{remark}

\section{Solution to the Robust Utility Maximization Problem}\label{sec:solution_to_the_robust_utility_maximization_problem}

In this section we solve~\eqref{eq:value_function_robust_power_constrained_time-dependent} by computing the optimal strategy $\pi^{(t),*}$ and the worst-case drift $\mu^{(t),*}$ and prove a minimax theorem in analogy to Theorem~\ref{thm:duality_result}. We proceed as in the setting with constant drift in Section~\ref{sec:recap_of_results_for_constant_drift}. Looking at the local optimization problem for a fixed $t\in[0,T]$ enables us to reduce the drift uncertainty to $K_t$ and make use of our results for constant drift. At the end of this section we explain which strategy will be realized by an investor whose information about the drift process changes continuously in time and how this strategy is naturally obtained from the solution to the local optimization problems.

\subsection{Solution to the non-robust problem}

As a first step towards solving~\eqref{eq:value_function_robust_power_constrained_time-dependent} we compute the optimal strategy for an investor given a particular future drift $\mu^{(t)}\in\calK^{(t)}$.
Due to the constraint on the admissible strategies, which prevents a pure bond investment, we reduce this problem to a less-dimensional unconstrained problem for which standard results from Merton~\cite{merton_1969} apply.

This is only the first step in solving the robust problem, but knowing the optimal strategy given a fixed drift will later enable us to compute the worst-case drift process and prove a minimax theorem.

\begin{proposition}\label{prop:optimal_strategy_non-robust_time-dependent}
	Let $t\in[0,T]$, $x>0$ and $\mu^{(t)}\in\calK^{(t)}$. Then the optimal strategy for the optimization problem
	\[ \sup_{\pi^{(t)}\in\calA_h(t,x)} \E_{\mu^{(t)}}\Bigl[U_\gamma\bigl(X^{t,x,\pi^{(t)}}_T\bigr)\Bigr] \]
	is the strategy $(\pi^{(t)}_s)_{s\in[t,T]}$ with
	\[ \pi^{(t)}_s  = \frac{1}{1-\gamma}A\mu^{(t)}_s+hc \]
	for all $s\in[t,T]$, where $A\in\R^{d\times d}$ and $c\in\R^d$ are as introduced in Definition~\ref{def:matrix_A_vector_c}.
\end{proposition}

\begin{proof}
	The proof works along the lines of the proof of Proposition~\ref{prop:optimal_strategy_non-robust}. We take an arbitrary strategy $\pi=\pi^{(t)}\in\calA_h(t,x)$ and recall that we can write the terminal wealth given $X_t^\pi =x$ under strategy $\pi$ as
	\[ X^{t,x,\pi}_T = x\,\exp\biggl(\int_t^T \Bigl(r+\pi_s^\transp(\mu^{(t)}_s-r\ones)-\frac{1}{2}\lVert\sigma^\transp\pi_s\rVert^2\Bigr)\rmd s + \int_t^T \pi_s^\transp\sigma\,\rmd W^\mu_s\biggr). \]
	We now proceed exactly as in the proof of Proposition~\ref{prop:optimal_strategy_non-robust}, replacing the constant $\mu$ by the $\calG_t$-measurable $(\mu^{(t)}_s)_{s\in[t,T]}$, and perform the same transformation to a $(d-1)$-dimensional unconstrained financial market.
	
	We can deduce that $\E_{\mu^{(t)}}[U_\gamma(X^{t,x,\pi}_T)]$ equals the expected utility of terminal wealth, conditional on $X_t^\pi=x$, in an unconstrained financial market with $d-1$ risky assets, where the future drift process is $(\widetilde{\mu}_s)_{s\in[t,T]}$, the risk-free interest rate is $(\widetilde{r}_s)_{s\in[t,T]}$ and the volatility matrix is $\widetilde{\sigma}\in\R^{(d-1)\times m}$. These transformed market parameters have the form
	\begin{align*}
		\widetilde{\sigma}&=D\sigma, \\
		\widetilde{r}_s&=(1-h)r+he_d^\transp\mu^{(t)}_s-\frac{1}{2}(1-\gamma)\lVert h\sigma^\transp e_d \rVert^2, \\
		\widetilde{\mu}_s&=D\mu^{(t)}_s - h(1-\gamma)D\sigma\sigma^\transp e_d+\widetilde{r}_s\mathbf{1}_{d-1}.
	\end{align*}
	Note that since the $(\mu^{(t)}_s)_{s\in[t,T]}$ are $\calG_t$-measurable, so are $(\widetilde{r}_s)_{s\in[t,T]}$ and $(\widetilde{\mu}_s)_{s\in[t,T]}$, in particular the market parameters in the transformed market can be observed by the investor.
	In this $(d-1)$-dimensional unconstrained financial market we can deduce from Merton~\cite{merton_1969} that the optimal strategy is of the form
	\begin{equation}\label{eq:optimal_pi_tilde_time-dependent}
		\widetilde{\pi}_s = \frac{1}{1-\gamma}(\widetilde{\sigma}\widetilde{\sigma}^\transp)^{-1}(\widetilde{\mu}_s-\widetilde{r}_s\mathbf{1}_{d-1})
		= \frac{1}{1-\gamma}(D\sigma\sigma^\transp D^\transp)^{-1}\bigl(D\mu^{(t)}_s - h(1-\gamma)D\sigma\sigma^\transp e_d\bigr)
	\end{equation}
	for every $s\in[t,T]$.
	For the logarithmic utility case, this is immediate, for power utility, this needs to be shown.
	
	Merton~\cite{merton_1969} yields the form of the optimal strategy in a Black--Scholes market with constant parameters. This result can be extended to a market where the risk-free interest rate as well as drift and volatility of the stocks are not necessarily constant but still observable by the investor, see Westphal~\cite[App.~B]{westphal_2019} for a complete proof. A similar result has been proven in Karatzas et al.~\cite{karatzas_lehoczky_shreve_xu_1991} for complete markets with deterministic market coefficients and for incomplete markets with totally unhedgeable market coefficients.
	
	Now we can return to our original market and obtain that the optimal strategy fulfills
	\begin{equation*}
		\begin{aligned}
			\pi^{(t)}_s &= D^\transp\widetilde{\pi}_s+he_d\\
			&= D^\transp\frac{1}{1-\gamma}(D\sigma\sigma^\transp D^\transp)^{-1}\bigl(D\mu^{(t)}_s - h(1-\gamma)D\sigma\sigma^\transp e_d\bigr)+he_d \\
			&= \frac{1}{1-\gamma}D^\transp(D\sigma\sigma^\transp D^\transp)^{-1}D\mu^{(t)}_s +h\bigl(I_d-D^\transp(D\sigma\sigma^\transp D^\transp)^{-1}D\sigma\sigma^\transp\bigr)e_d \\
			&= \frac{1}{1-\gamma}A\mu^{(t)}_s+hc
		\end{aligned}
	\end{equation*}
	for all $s\in[t,T]$, where we have used the notation for $A$ and $c$ from Definition~\ref{def:matrix_A_vector_c}. Note that $(\pi^{(t)}_s)_{s\in[t,T]}$ is indeed admissible due to boundedness of $K_t$.
\end{proof}

The preceding proposition states the form of the investor's optimal strategy under the assumption that a specific future drift process $(\mu^{(t)}_s)_{s\in[t,T]}$ is given. The explicit form can be used to compute also the expected utility obtained when applying the optimal strategy.

\begin{corollary}
	Let $t\in[0,T]$ and $\mu^{(t)}\in\calK^{(t)}$. Then the optimal expected utility from terminal wealth is
	\begin{equation*}
		\begin{aligned}
			&\sup_{\pi^{(t)}\in\calA_h(t,x)} \E_{\mu^{(t)}}\Bigl[U_\gamma\bigl(X^{t,x,\pi^{(t)}}_T\bigr)\Bigr]\\
			&\quad=
			\begin{dcases}
				\frac{x^\gamma}{\gamma}\exp\biggl(\gamma\! \int_t^T\!\Bigl( \widetilde{r}_s+\frac{1}{2(1-\gamma)}\bigl(\widetilde{\mu}_s-\widetilde{r}_s\mathbf{1}_{d-1}\bigr)^\transp(\widetilde{\sigma}\widetilde{\sigma}^\transp )^{-1}\bigl(\widetilde{\mu}_s-\widetilde{r}_s\mathbf{1}_{d-1}\bigr)\Bigr)\rmd s\!\biggr), &\gamma\neq 0,\\
				\log(x) + \int_t^T\Bigl( \widetilde{r}_s+\frac{1}{2}\bigl(\widetilde{\mu}_s-\widetilde{r}_s\mathbf{1}_{d-1}\bigr)^\transp(\widetilde{\sigma}\widetilde{\sigma}^\transp )^{-1}\bigl(\widetilde{\mu}_s-\widetilde{r}_s\mathbf{1}_{d-1}\bigr) \Bigr)\rmd s, &\gamma=0,
			\end{dcases}
		\end{aligned}
	\end{equation*}
	where
	\begin{equation*}
		\begin{aligned}
			\widetilde{\sigma}&=D\sigma, \\
			\widetilde{r}_s&=(1-h)r+he_d^\transp\mu^{(t)}_s-\frac{1}{2}(1-\gamma)\lVert h\sigma^\transp e_d \rVert^2, \\
			\widetilde{\mu}_s&=D\mu^{(t)}_s - h(1-\gamma)D\sigma\sigma^\transp e_d+\widetilde{r}_s\mathbf{1}_{d-1}.
		\end{aligned}
	\end{equation*}
\end{corollary}

\begin{proof}
	The representation in the corollary follows, just like in the proof of Corollary~\ref{cor:optimal_utility_non-robust}, by the fact that we have reduced our constrained utility maximization problem to a $(d-1)$-dimensional unconstrained problem where the parameters of our transformed financial market are exactly those that are listed in the corollary.
	We have seen that the optimal strategy in this $(d-1)$-dimensional market fulfills
	\[ \widetilde{\pi}_s = \frac{1}{1-\gamma}(\widetilde{\sigma}\widetilde{\sigma}^\transp)^{-1}(\widetilde{\mu}_s-\widetilde{r}_s\mathbf{1}_{d-1}) \]
	for all $s\in[t,T]$. Plugging this optimal strategy in yields the expression from the corollary.
\end{proof}

\subsection{The worst-case drift process}

In the following, we compute the worst-case future drift process that is determined at time $t\in[0,T]$, i.e.\ the drift process $\mu^{(t)}\in\calK^{(t)}$ for which
\[ \sup_{\pi^{(t)}\in\calA_h(t,x)} \E_{\mu^{(t)}}\Bigl[U_\gamma\bigl(X^{t,x,\pi^{(t)}}_T\bigr)\Bigr] \]
is minimized. Due to the previous corollary we see that this is equivalent to the minimization of the integral
\begin{equation}\label{eq:what_has_to_be_minimized_time-dependent}
	\int_t^T\Bigl( \widetilde{r}_s+\frac{1}{2}\bigl(\widetilde{\mu}_s-\widetilde{r}_s\mathbf{1}_{d-1}\bigr)^\transp(\widetilde{\sigma}\widetilde{\sigma}^\transp )^{-1}\bigl(\widetilde{\mu}_s-\widetilde{r}_s\mathbf{1}_{d-1}\bigr) \Bigr)\rmd s.
\end{equation}
When plugging the representations for $\widetilde{\mu}$, $\widetilde{r}$ and $\widetilde{\sigma}$ back in, we obtain an expression that depends on $(\mu^{(t)}_s)_{s\in[t,T]}$ again. By the same calculations as in the setting with constant drift we deduce that minimizing~\eqref{eq:what_has_to_be_minimized_time-dependent} is equivalent to minimizing
\[ \int_t^T\Bigl(\frac{1}{2(1-\gamma)}(\mu^{(t)}_s)^\transp A\mu^{(t)}_s+hc^\transp\mu^{(t)}_s \Bigr)\rmd s. \]
But the minimization of this integral is equivalent to a pointwise minimization of
\[ K_t\ni\mu \mapsto \frac{1}{2(1-\gamma)}\mu^\transp A\mu+hc^\transp\mu. \]
Now it is straightforward to see that we can use the results from Section~\ref{sec:recap_of_results_for_constant_drift} to obtain the worst-case drift process $(\mu^{(t),*}_s)_{s\in[t,T]}$. Here, $\mu^{(t),*}_s$ is for any $s\in[t,T]$ obtained as the minimizer of the above function on $K_t$. Recall that the uncertainty set is an ellipsoid of the form
\[ K_t = \{ \mu\in\R^d \,|\, (\mu-\nu_t)^\transp\Gamma_t^{-1}(\mu-\nu_t)\leq \kappa_t^2 \}. \]
We have assumed that $\Gamma_t$ is a symmetric positive-definite matrix in $\R^{d\times d}$. In the following we use the representation $\Gamma_t=\tau_t\tau_t^\transp$ where $\tau_t\in\R^{d\times d}$ is a nonsingular matrix.

\begin{corollary}\label{cor:solution_of_the_inf_sup_problem_time-dependent}
	We fix some $t\in[0,T]$ and let $0=\lambda_{t,1}<\lambda_{t,2}\leq\cdots\leq\lambda_{t,d}$ denote the eigenvalues of $\tau_t^\transp A\tau_t$, and
	\[ v_{t,1}=\frac{1}{\lVert \tau_t^{-1}\ones \rVert}\tau_t^{-1}\ones, v_{t,2},\dots,v_{t,d}\in\R^d \]
	the respective orthogonal eigenvectors with $\lVert v_{t,i}\rVert=1$ for all $i=1,\dots, d$.
	Then
	\[ \adjustlimits \inf_{\mu^{(t)}\in \calK^{(t)}} \sup_{\pi^{(t)}\in\calA_h(t,x)} \E_{\mu^{(t)}}\Bigl[U_\gamma\bigl(X^{t,x,\pi^{(t)}}_T\bigr)\Bigr] = \E_{\mu^{(t),*}}\Bigl[U_\gamma\bigl(X^{t,x,\pi^{(t),*}}_T\bigr)\Bigr], \]
	where
	\[ \mu^{(t),*}_s=\nu_t-\tau_t\sum_{i=1}^d \biggl(\frac{\lambda_{t,i}}{1-\gamma}+\frac{h}{\psi_t(\kappa_t)\lVert \tau_t^{-1}\ones \rVert}\biggr)^{-1}\Bigl\langle h\tau_t^\transp c+\frac{\lambda_{t,i}}{1-\gamma}\tau_t^{-1}\nu_t, v_{t,i}\Bigr\rangle\, v_{t,i} \]
	for all $s\in[t,T]$, and where $\psi_t(\kappa_t)\in(0,\kappa_t]$ is uniquely determined by $\lVert\tau_t^{-1}(\mu^{(t),*}_s-\nu_t)\rVert=\kappa_t$. The strategy $(\pi^{(t),*}_s)_{s\in[t,T]}$ has the form
	\[ \pi^{(t),*}_s = \frac{1}{1-\gamma}A\mu^{(t),*}_s +hc \]
	for all $s\in[t,T]$.
\end{corollary}

\begin{proof}
	We have seen that the worst-case drift process $(\mu^{(t),*}_s)_{s\in[t,T]}$ is the one where $\mu^{(t),*}_s$ is for any $s\in[t,T]$ equal to the minimizer of the function
	\[ \mu \mapsto \frac{1}{2(1-\gamma)}\mu^\transp A\mu+hc^\transp\mu \]
	over all $\mu\in K_t$. So we can do the minimization as in Section~\ref{sec:recap_of_results_for_constant_drift}. We know that the matrix $\tau_t^\transp A\tau_t\in\R^{d\times d}$ is symmetric and positive definite with
	\[ \mathrm{ker}(\tau_t^\transp A\tau_t)=\mathrm{span}(\{\tau_t^{-1}\ones\}). \]
	Now the representation of $\mu^{(t),*}_s$ follows as in Theorem~\ref{thm:solution_of_the_inf_sup_problem}. The form of the optimal strategy $\pi^{(t),*}$ then follows from Proposition~\ref{prop:optimal_strategy_non-robust_time-dependent}.
\end{proof}

The preceding corollary shows that the problem
\[ \adjustlimits \inf_{\mu^{(t)}\in \calK^{(t)}} \sup_{\pi^{(t)}\in\calA_h(t,x)} \E_{\mu^{(t)}}\Bigl[U_\gamma\bigl(X^{t,x,\pi^{(t)}}_T\bigr)\Bigr] \]
is solved by drift process $(\mu^{(t),*}_s)_{s\in[t,T]}$ and strategy $(\pi^{(t),*}_s)_{s\in[t,T]}$. Note that both the worst-case drift process and the optimal strategy are constant on $[t,T]$ and $\calG_t$-measurable. This is due to the setup of the model in which investors assume that the future drift process will take values in the ellipsoid $K_t$ only.

The problem above is the dual to our original problem
\[ \adjustlimits \sup_{\pi^{(t)}\in\calA_h(t,x)} \inf_{\mu^{(t)}\in \calK^{(t)}} \E_{\mu^{(t)}}\Bigl[U_\gamma\bigl(X^{t,x,\pi^{(t)}}_T\bigr)\Bigr]. \]
To ensure that $\mu^{(t),*}$ and $\pi^{(t),*}$ are also a solution to this problem we have to show that $\mu^{(t),*}$ is the worst drift process in the set $\calK^{(t)}$, given that an investor chooses trading strategy $\pi^{(t),*}$. In that case, the infimum and the supremum interchange and we can deduce that $\pi^{(t),*}$ and $\mu^{(t),*}$ also establish a solution to our original robust optimization problem.

\subsection{A minimax theorem}

We proceed as in Section~\ref{sec:recap_of_results_for_constant_drift} and note that the strategy $\pi^{(t),*}$ from the previous corollary satisfies
\[ \pi^{(t),*}_s = -\frac{h}{\psi_t(\kappa_t)\lVert\tau_t^{-1}\ones\rVert}\Gamma_t^{-1}\bigl(\mu^{(t),*}_s-\nu_t\bigr) \]
for all $s\in[t,T]$. This can be proven by analogy with Lemma~\ref{lem:representation_of_pi_star}. This observation helps to prove the following proposition.

\begin{proposition}
	The drift process $(\mu^{(t)}_s)_{s\in[t,T]}$ that attains for any $x>0$ the minimum in
	\[ \inf_{\mu^{(t)}\in \calK^{(t)}} \E_{\mu^{(t)}}\Bigl[U_\gamma\bigl(X^{t,x,\pi^{(t),*}}_T\bigr)\Bigr] \]
	is $(\mu^{(t),*}_s)_{s\in[t,T]}$, i.e.\ $\mu^{(t),*}$ is the worst possible drift process, given that an investor chooses the strategy $\pi^{(t),*}$.
\end{proposition}

\begin{proof}
	We take an arbitrary $\mu=\mu^{(t)}\in\calK^{(t)}$. Note that in case $\gamma\neq 0$ we can write
	\begin{equation*}
		\begin{aligned}
			&\E_{\mu}\Bigl[U_\gamma\bigl(X^{t,x,\pi^{(t),*}}_T\bigr)\Bigr]\\
			&= \frac{x^\gamma}{\gamma}\rme^{\gamma r(T-t)}\E_{\mu}\biggl[\exp\biggl(\!\gamma\!\int_t^T\!\!\! \Bigl((\pi^{(t),*}_s)\!^\transp\!(\mu_s-r\ones) -\frac{1}{2}\lVert\sigma^\transp\!\pi^{(t),*}_s\rVert^2\Bigr)\rmd s + \!\gamma\!\int_t^T \!\!\!(\pi^{(t),*}_s)\!^\transp\! \sigma\,\rmd W^\mu_s \!\biggr)\biggr] \\
			&= \frac{x^\gamma}{\gamma}\rme^{\gamma r(T-t)}\exp\biggl(\gamma\int_t^T \Bigl((\pi^{(t),*}_s)^\transp(\mu_s-r\ones) -\frac{1-\gamma}{2}\lVert\sigma^\transp\pi^{(t),*}_s\rVert^2\Bigr)\rmd s\biggl).
		\end{aligned}
	\end{equation*}
	In case $\gamma=0$ we have
	\[ \E_{\mu}\Bigl[\log\bigl(X^{t,x,\pi^{(t),*}}_T\bigr)\Bigr]
	=\log(x)+ r(T-t)+\int_t^T \Bigl((\pi^{(t),*}_s)^\transp(\mu_s-r\ones) -\frac{1}{2}\lVert\sigma^\transp\pi^{(t),*}_s\rVert^2\Bigr)\rmd s. \]
	In both cases, the drift process $(\mu_s)_{s\in[t,T]}\in \calK^{(t)}$ that minimizes this expression is the one that minimizes
	\[ \int_t^T (\pi^{(t),*}_s)^\transp\mu_s\,\rmd s. \]
	Since $(\pi^{(t),*}_s)_{s\in[t,T]}$ is constant, we find the minimizer as the minimizer of $(\pi^{(t),*}_s)^\transp\mu_s$. Recall that
	\[ \pi^{(t),*}_s = -\frac{h}{\psi_t(\kappa_t)\lVert\tau_t^{-1}\ones\rVert}\Gamma_t^{-1}\bigl(\mu^{(t),*}_s-\nu_t\bigr). \]
	It follows that
	\[ (\pi^{(t),*}_s)^\transp\Gamma_t\pi^{(t),*}_s = \frac{h^2}{\psi_t(\kappa_t)^2\lVert\tau_t^{-1}\ones\rVert^2}\bigl(\mu^{(t),*}_s-\nu_t\bigr)^\transp\Gamma_t^{-1}\bigl(\mu^{(t),*}_s-\nu_t\bigr) = \frac{h^2\kappa_t^2}{\psi_t(\kappa_t)^2\lVert\tau_t^{-1}\ones\rVert^2}. \]
	Knowing that $\psi_t(\kappa_t)>0$ we can deduce
	\[ \sqrt{(\pi^{(t),*}_s)^\transp\Gamma_t\pi^{(t),*}_s}=\frac{h\kappa_t}{\psi_t(\kappa_t)\lVert\tau_t^{-1}\ones\rVert}. \]
	The drift process $\mu^{(t),*}_s$ at time $s$ can thus be rewritten in the form
	\[ \mu^{(t),*}_s=\nu_t-\frac{\psi_t(\kappa_t)\lVert\tau_t^{-1}\ones\rVert}{h}\Gamma_t\pi^{(t),*}_s = \nu_t-\frac{\kappa_t}{\sqrt{(\pi^{(t),*}_s)^\transp\Gamma_t\pi^{(t),*}_s}}\Gamma_t\pi^{(t),*}_s. \]
	This is exactly the vector that minimizes $(\pi^{(t),*}_s)^\transp\mu$ over all $\mu\in K_t$, see the proof of \cite[Prop.~3.11]{sass_westphal_2020}.
	Hence, $\mu^{(t),*}$ is the drift process that minimizes the expected utility of terminal wealth for an investor who chooses strategy $\pi^{(t),*}$.
\end{proof}

The previous proposition establishes an equilibrium result. By definition, the strategy $\pi^{(t),*}$ is optimal for the drift $\mu^{(t),*}$. Due to the proposition, it also holds that $\mu^{(t),*}$ is the worst drift given that an investor chooses strategy $\pi^{(t),*}$. Hence, we see that $(\pi^{(t),*},\mu^{(t),*})$ is a saddle point of the optimization problem
\[ \adjustlimits \sup_{\pi^{(t)}\in\calA_h(t,x)} \inf_{\mu^{(t)}\in \calK^{(t)}} \E_{\mu^{(t)}}\Bigl[U_\gamma\bigl(X^{t,x,\pi^{(t)}}_T\bigr)\Bigr]. \]
In particular, the supremum and infimum can be interchanged. We obtain the following minimax theorem.

\begin{theorem}\label{thm:minimax_theorem}
	Let $t\in[0,T]$. Then for $x>0$
	\begin{equation*}
		\begin{aligned}
			\adjustlimits \sup_{\pi^{(t)}\in\calA_h(t,x)} \inf_{\mu^{(t)}\in \calK^{(t)}} \E_{\mu^{(t)}}\Bigl[U_\gamma\bigl(X^{t,x,\pi^{(t)}}_T\bigr)\Bigr]
			&= \E_{\mu^{(t),*}}\Bigl[U_\gamma\bigl(X^{t,x,\pi^{(t),*}}_T\bigr)\Bigr] \\
			&= \adjustlimits \inf_{\mu^{(t)}\in \calK^{(t)}} \sup_{\pi^{(t)}\in\calA_h(t,x)} \E_{\mu^{(t),*}}\Bigl[U_\gamma\bigl(X^{t,x,\pi^{(t),*}}_T\bigr)\Bigr],
		\end{aligned}
	\end{equation*}
	where $\mu^{(t),*}$ and $\pi^{(t),*}$ are defined as in Corollary~\ref{cor:solution_of_the_inf_sup_problem_time-dependent}.
\end{theorem}

\begin{proof}
	The proof is analogous to the proof of Theorem~\ref{thm:duality_result}.
\end{proof}

The previous theorem solves our original local optimization problem~\eqref{eq:value_function_robust_power_constrained_time-dependent} for a fixed time $t\in[0,T]$. It shows that the best strategy for an investor in this robust optimization problem is the strategy $(\pi^{(t),*}_s)_{s\in[t,T]}$ with
\[ \pi^{(t),*}_s=\frac{1}{1-\gamma}A\mu^{(t),*}_s+hc \]
for all $s\in[t,T]$, where $(\mu^{(t),*}_s)_{s\in[t,T]}$ is defined as in Corollary~\ref{cor:solution_of_the_inf_sup_problem_time-dependent}. The process $(\mu^{(t),*}_s)_{s\in[t,T]}$ can be interpreted as the worst possible realization of the future drift process from the investor's point of view at time $t$. The worst-case drift and optimal strategy in this setting are constant on $[t,T]$. Since these do not depend on $X_t^\pi=x$, this would also be true for the unconditional case. 
This is due to the assumption of the investor that the future drift will take values in the set $K_t$ only, where $K_t$ is determined at time $t$ using all available information, i.e.\ $K_t$ is $\calG_t$-measurable.

In our continuous-time setting it is likely that the information about the unobservable true drift process changes continuously, therefore also the uncertainty set $K_t$ will be updated continuously in time. At each time $t\in[0,T]$, the investor will revise both the uncertainty set and the optimization problem
\[ \adjustlimits \sup_{\pi^{(t)}\in\calA_h(t,x)} \inf_{\mu^{(t)}\in \calK^{(t)}} \E_{\mu^{(t)}}\Bigl[U_\gamma\bigl(X^{t,x,\pi^{(t)}}_T\bigr)\Bigr]. \]
The strategy that is realized by the investor can then be found as $(\pi^*_t)_{t\in[0,T]}$ with $\pi^*_t=\pi^{(t),*}_t$ for any $t\in[0,T]$. It has the form
\[ \pi^*_t=\frac{1}{1-\gamma}A\mu^*_t +hc \]
where $(\mu^*_t)_{t\in[0,T]}$ is constructed via $\mu^*_t=\mu^{(t),*}_t$ for all $t\in[0,T]$. Note that the processes $(\mu^*_t)_{t\in[0,T]}$ and $(\pi^*_t)_{t\in[0,T]}$ are uniquely determined, $\mathbb{G}$-adapted and in general non-constant.

In the special case where $K_t = K_0$ for all $t\in[0,T]$, i.e.\ where our reference drift is simply a constant $\nu$, and also the matrix $\Gamma_t=\Gamma$ as well as the radius $\kappa_t=\kappa$ are constant in time, also $(\mu^*_t)_{t\in[0,T]}$ and $(\pi^*_t)_{t\in[0,T]}$ are constant in time. The constant values are the ones that we also get in the setting with constant drift and uncertainty set in Theorem~\ref{thm:solution_of_the_inf_sup_problem}.

\section{Construction of Uncertainty Sets via Filters}\label{sec:construction_of_uncertainty_sets_via_filters}

In the preceding sections we have seen how the duality approach from Sass and Westphal~\cite{sass_westphal_2020} carries over to a financial market where the drift is not necessarily constant. The generalized model allows for local uncertainty sets of the form
\[ K_t = \bigl\{ \mu\in\R^d \,\big|\, (\mu-\nu_t)^\transp\Gamma_t^{-1}(\mu-\nu_t)\leq \kappa_t^2 \bigr\}, \quad t\in[0,T]. \]
We have fixed an investor filtration $\mathbb{G}=(\calG_t)_{t\in[0,T]}$ describing the investor's information in the course of time. Our model then assumes that the processes $\nu=(\nu_t)_{t\in[0,T]}$, $\Gamma=(\Gamma_t)_{t\in[0,T]}$ and $\kappa=(\kappa_t)_{t\in[0,T]}$ are $\mathbb{G}$-adapted. Recall that $\nu$ takes values in $\R^d$, $\Gamma$ in the set of symmetric and positive-definite matrices in $\R^{d\times d}$ and $\kappa$ on the positive real line.
We motivated the reference drift $\nu$ as an estimation for the true drift, based on the information available to the investor. Here we want to make this more specific by considering the filter.

\subsection{Confidence regions as uncertainty sets}

The filter is the conditional distribution of $\mu$ given the available information $\mathbb{G}$. We take $\nu$ to be the conditional expectation of the drift given $\mathbb{G}$, i.e.\ $\nu_t=\muhat{}{t}:=\E[\mu_t\,|\,\calG_t]$ for every $t\in[0,T]$. The conditional covariance matrix
\[ \gam{}{t} := \E\bigl[(\mu_t-\muhat{}{t})(\mu_t-\muhat{}{t})^\transp\,\big|\,\calG_t\bigr] \]
measures how close the estimator $\muhat{}{t}$ is to the true drift. Note that by construction both $(\muhat{}{t})_{t\in[0,T]}$ and $(\gam{}{t})_{t\in[0,T]}$ are $\mathbb{G}$-adapted processes. The key idea for constructing uncertainty sets based on the filter is to create confidence regions centered around $\muhat{}{t}$, shaped by $\gam{}{t}$ for every $t\in[0,T]$.

Let us assume that the drift process and the investor filtration are such that the filter is normally distributed, more precisely
\[ \mu_t\,|\,\calG_t \sim \calN(\muhat{}{t},\gam{}{t}). \]
By applying a simple transformation we deduce that
\[ (\mu_t-\muhat{}{t})^\transp(\gam{}{t})^{-1}(\mu_t-\muhat{}{t}) \]
given $\calG_t$ is $\chi^2$-distributed with $d$ degrees of freedom. We fix some $\eta\in(0,1)$ and observe that a $(1-\eta)$-confidence region can be obtained from
\[ 1-\eta = \mathbb{P}\Bigl((\mu_t-\muhat{}{t})^\transp(\gam{}{t})^{-1}(\mu_t-\muhat{}{t}) \leq \chi^2_{d,1-\eta} \,\Big|\,\calG_t\Bigr). \]
Here, $\chi^2_{d,1-\eta}$ denotes the $(1-\eta)$-quantile of the $\chi^2$-distribution with $d$ degrees of freedom. This motivates the choice of
\[ K_t = \bigl\{ \mu\in\R^d \,\big|\, (\mu-\muhat{}{t})^\transp(\gam{}{t})^{-1}(\mu-\muhat{}{t}) \leq \chi^2_{d,1-\eta} \bigr\}, \quad t\in[0,T], \]
i.e.\ taking $\nu_t=\muhat{}{t}$, $\Gamma_t=\gam{}{t}$ and $\kappa_t=\sqrt{\chi^2_{d,1-\eta}}$ for every $t\in[0,T]$.

If indeed $\mu_t$ given $\calG_t$ is normally distributed, we additionally know that at any fixed time $t\in[0,T]$ the probability that $\mu_t\in K_t$, conditional on $\calG_t$, is equal to $1-\eta$. Note that $K_t$ is still a reasonable uncertainty set for $\mu_t$ in the case where the assumption about the normal distribution of the filter is not fulfilled.

\subsection{Comparison of different investor filtrations}

The preceding section explains how time-dependent uncertainty sets can be created based on filters. We now apply this to a model with an unobservable Ornstein--Uhlenbeck drift process and unbiased, normally distributed expert opinions arriving at discrete points in time. The setting is based on Gabih et al.~\cite{gabih_kondakji_sass_wunderlich_2014} as well as Sass et al.~\cite{sass_westphal_wunderlich_2017, sass_westphal_wunderlich_2021}.
Returns in this setting are modelled as
\[ \rmd R_t = \mu_t\,\rmd t+\sigma_R\,\rmd W^R_t, \]
where $W^R=(W^R_t)_{t\in[0,T]}$ is an $m$-dimensional Brownian motion with $m\geq d$ and where we assume that $\sigma_R\in\R^{d\times m}$ has full rank. The drift process $\mu$ is defined by the Ornstein--Uhlenbeck dynamics
\[ \rmd \mu_t = \alpha (\delta - \mu_t)\,\rmd t + \beta\,\rmd B_t, \]
where $\alpha$ and $\beta\in\R^{d\times d}$, $\delta\in\R^d$ and $B=(B_t)_{t\in [0,T]}$ is a $d$-dimensional Brownian motion that is independent of $W^R$. The matrices $\alpha$ and $\beta\beta^\transp$ are assumed to be symmetric and positive definite. We further make the assumption that $\mu_0\sim\calN(m_0,\Sigma_0)$ for some $m_0\in\R^d$ and some symmetric and positive-semidefinite matrix $\Sigma_0\in\R^{d\times d}$, and that $\mu_0$ is independent of the Brownian motions $W^R$ and $B$, i.e.\ $\mu$ is independent of $W^R$.

Information about the drift process can be drawn from return observations. An additional source of information in this model are expert opinions that arrive at discrete points in time and give an unbiased estimate of the state of the drift at that time point.
We assume that the expert opinions arrive at the information dates $(T_k)_{k\in I}$ and that an expert opinion at time $T_k$ is of the form
\[ Z_k = \mu_{T_k}+(\Gamma_k)^{1/2}\varepsilon_k, \]
where the matrices $\Gamma_k\in\R^{d\times d}$ are symmetric and positive definite and the $\varepsilon_k$ are multivariate $\calN(0,I_d)$-distributed and independent of the Brownian motions in the market and of~$\mu_0$. The sequence of information dates $(T_k)_{k\in I}$ is also independent of the $(\varepsilon_k)_{k\in I}$ and the Brownian motions as well as of $\mu_0$. In particular, given $\mu_{T_k}$ the expert opinion is multivariate $\calN(\mu_{T_k},\Gamma_k)$-distributed.
The model then gives rise to various investor filtrations $\mathbb{G}=(\calG_t)_{t\in[0,T]}$. We consider the cases
\begin{alignat*}{3}
	&\mathbb{G}=\mathbb{F}^N && =(\calF^N_t)_{t\in[0,T]} && \text{ where } \calF^N_t=\sigma(\calN_{\mathbb{P}}), \\
	&\mathbb{G}=\mathbb{F}^R && =(\calF^R_t)_{t\in[0,T]} && \text{ where } \calF^R_t=\sigma((R_s)_{s\in[0,t]})\vee\sigma(\calN_{\mathbb{P}}), \\
	&\mathbb{G}=\mathbb{F}^E && =(\calF^E_t)_{t\in[0,T]} && \text{ where } \calF^E_t=\sigma((T_k,Z_k)_{T_k\leq t})\vee\sigma(\calN_{\mathbb{P}}), \\
	&\mathbb{G}=\mathbb{F}^C && =(\calF^C_t)_{t\in[0,T]} && \text{ where } \calF^C_t=\sigma((R_s)_{s\in[0,t]})\vee\sigma((T_k,Z_k)_{T_k\leq t})\vee\sigma(\calN_{\mathbb{P}})
\end{alignat*}
for the investor filtrations, where we write $\calN_{\mathbb{P}}$ for the set of null sets under $\mathbb{P}$, i.e.\ we work with the filtrations that are augmented by null sets. We speak of the investor with filtration $\mathbb{F}^H$, $H\in\{N,R,E,C\}$, as the $H$-investor. Note that the $N$-investor observes neither returns nor expert opinions and only has knowledge about the market parameters. The $R$-investor observes only the return process, the $E$-investor only the discrete-time expert opinions, and the $C$-investor the combination of both.

\begin{example}
	Based on one realization of the model's stochastic processes, fixing one information setting $H\in\{N,R,E,C\}$, we obtain one realization of the filter, leading to a time-dependent uncertainty set $K^H$. For illustration purposes we plot in Figure~\ref{fig:uncertainty_sets_various_investor_filtrations} against time a realization of the different filters with resulting uncertainty sets in a market with $d=1$ stock.
	
	The various subplots are all based on the same realization of the drift process $\mu$, returns $R$ and expert opinions $Z_k$. As a first case we consider in Figure~\ref{subf:H=E_n=0} the degenerate information setting $H=N$, corresponding to an investor who observes neither the return process nor the expert opinions. The only knowledge the investor has are the model parameters. The conditional mean is in this case constantly equal to the long-term mean $\delta$ of the drift process. The resulting uncertainty set converges very fast to a fixed interval centered around $\delta$.
	
	For $H=R$, the uncertainty set moves up and down along with the conditional mean as can be seen in Figure~\ref{subf:H=R}. In Figures~\ref{subf:H=E_n=10} and~\ref{subf:H=C_n=10} we have equidistant information dates with expert opinions. The corresponding uncertainty set jumps at information dates along with the conditional mean, due to the updates caused by an incoming expert opinion. In the case $H=E$ shown in Figure~\ref{subf:H=E_n=10}, no further information between the arrival times of the expert opinions is used. Therefore, the filter evolves deterministically in direction of the long-term mean between the jump times. In the case $H=C$ in Figure~\ref{subf:H=C_n=10}, in addition to the expert opinions  the return observations are used.  It also becomes apparent from the plots that the conditional variance decreases at information dates, leading to a shrinking uncertainty set.
	
	Neither of the information filtrations leads to a perfect uncertainty set in the sense that the true drift stays in that uncertainty set at any point in time. By the setup of the uncertainty set there is always a positive probability that the true drift process moves out of the uncertainty set at some point in time.
	
	 \begin{figure}[ht]
	 	\begin{subfigure}{.44\textwidth}
	 		\includegraphics{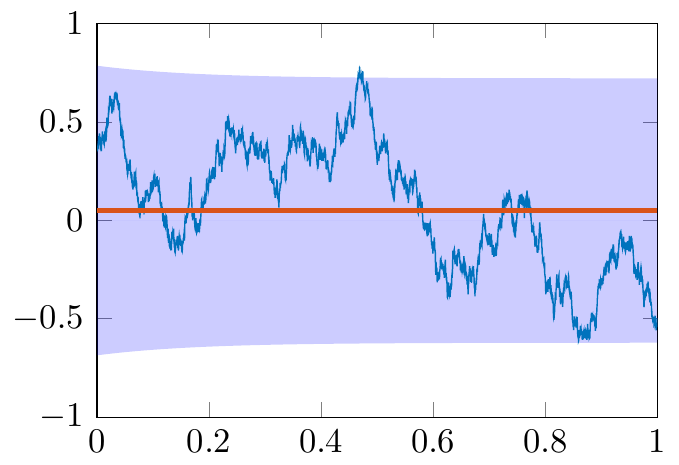}
	 		\caption{$\mathbb{G}=\mathbb{F}^N$}\label{subf:H=E_n=0}
	 	\end{subfigure}%
	 	\begin{subfigure}{.44\textwidth}
	 		\includegraphics{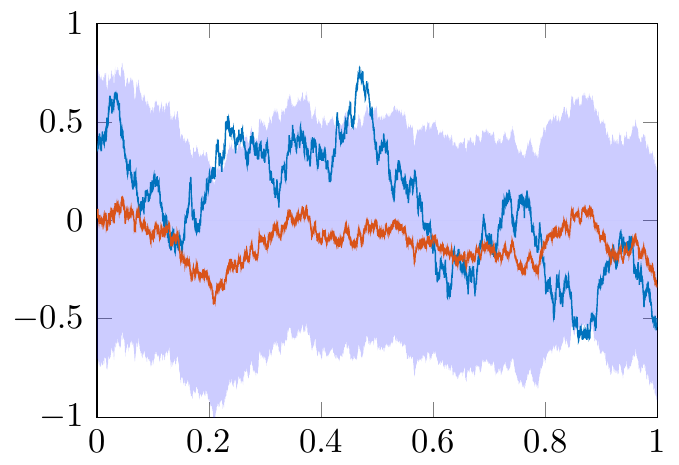}
	 		\caption{$\mathbb{G}=\mathbb{F}^R$}\label{subf:H=R}
	 	\end{subfigure}%
	 	\begin{subfigure}{.12\textwidth}
	 		\centering
	 		\includegraphics{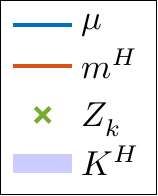}
	 	\end{subfigure}
	 	\newline
	 	\begin{subfigure}{.44\textwidth}
	 		\includegraphics{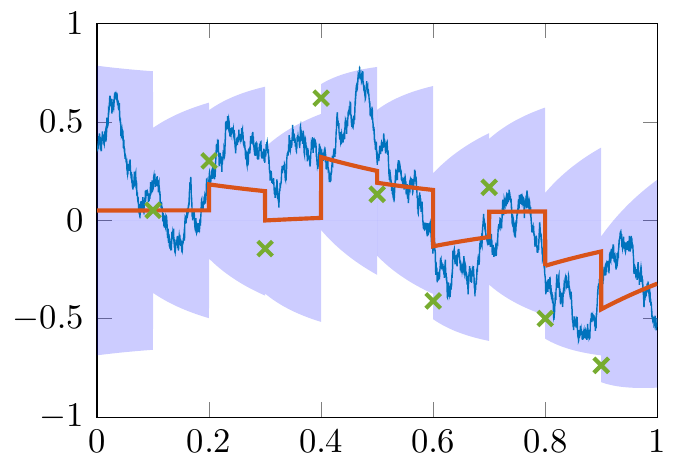}
	 		\caption{$\mathbb{G}=\mathbb{F}^E$}\label{subf:H=E_n=10}
	 	\end{subfigure}%
	 	\begin{subfigure}{.44\textwidth}
	 		\includegraphics{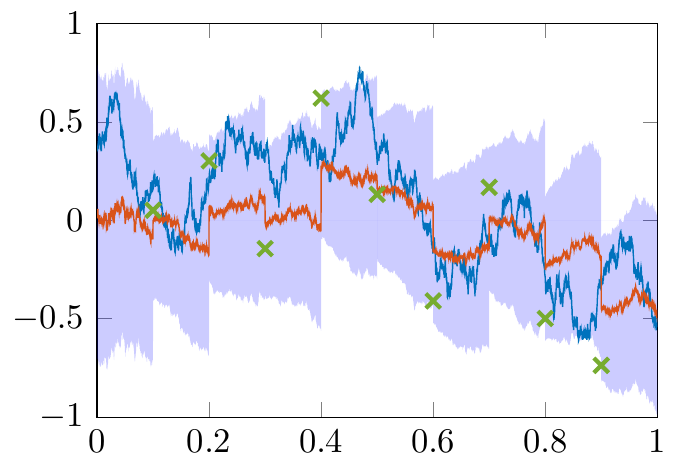}
	 		\caption{$\mathbb{G}=\mathbb{F}^C$}\label{subf:H=C_n=10}
	 	\end{subfigure}
	 	
	 	\caption{Uncertainty sets based on filters for various investor filtrations $\mathbb{F}^H$. Each subplot is based on the same realization of the drift and return process and expert opinions. Based on this realization, the filter of the $H$-investor can be computed. The uncertainty set $K^H$ is then determined according to the filter realization.}
	 	\label{fig:uncertainty_sets_various_investor_filtrations}
	 \end{figure}

%
\end{example}

We now give a numerical example to illustrate the effect that the worst-case optimization among uncertainty sets created from filters has for the various investor filtrations considered before. We create for a fixed realization of the drift process $\mu$, of the return process $R$ and the expert opinions $Z_k$ a time-dependent uncertainty set for each of the corresponding filters. The aim is to compare the robust strategies that take into account model uncertainty with the ``naive'' strategies that rely on the respective drift estimates, only.

\paragraph{Model parameters.}
We want to apply our worst-case utility maximization problem, in particular also imposing the constraint $\langle\pi_t,\ones\rangle=h$ on the investor's strategies. For that purpose we take a market with $d=2$ stocks here. We fix an investment horizon of $T=1$ and take $h=1$. Moreover, we assume that investors start with an initial wealth of $x_0=1$, use power utility functions $U_\gamma$ with $\gamma=0.5$ and a confidence level $\eta=0.1$ to create their uncertainty sets. Further parameters of the market are given in Table~\ref{tab:market_parameters_for_utility_study}.

\begin{table}[ht]
	\centering
	\begin{tabular}{llll}
		\hline
		\addlinespace[1ex]
		mean reversion speed of drift process	&$\alpha$	&$=$	&$\begin{pmatrix}
										 3 & 0\\
										 0 & 2
										 \end{pmatrix}$ \\
		\addlinespace[1ex]
		volatility of drift process		&$\beta$	&$=$	&$\begin{pmatrix}
										 0.50 & 0.25\\
										 0.25 & 0.50
										 \end{pmatrix}$ \\
		\addlinespace[1ex]
		mean reversion level of drift process	&$\delta$	&$=$	&$\begin{pmatrix}
										 0.02\\
										 0.03
										 \end{pmatrix}$ \\
		\addlinespace[1ex]
		initial mean of drift process		&$m_0$		&$=$	&$\begin{pmatrix}
										 0.02\\
										 0.03
										 \end{pmatrix}$ \\
		\addlinespace[1ex]
		initial variance of drift process	&$\Sigma_0$	&$=$	&$\begin{pmatrix}
										 0.01 & 0\\
										 0 & 0.01
										 \end{pmatrix}$ \\
		\addlinespace[1ex]
		volatility of returns			&$\sigma_R$	&$=$	&$\begin{pmatrix}
										 0.10 & 0.05\\
										 0.05 & 0.01
										 \end{pmatrix}$ \\
		\addlinespace[1ex]
		volatility of continuous expert		&$\sigma_J$	&$=$	&$\begin{pmatrix}
										 0.10 & 0.05\\
										 0.05 & 0.01
										 \end{pmatrix}$ \\
		\addlinespace[1ex]\hline
	\end{tabular}
	\caption{Market parameters for numerical example.}\label{tab:market_parameters_for_utility_study}
\end{table}

\paragraph{Simulation study.}
For the given model parameters we simulate a drift process, the return process $R$ and $n=10$ discrete-time expert opinions arriving at deterministic and equidistant information dates on $[0,T]$. We then obtain a realization of the filters $(\muhat{H}{}, \gam{H}{})$ for any of the information settings $H$ from above. As before, this leads to one time-dependent uncertainty set for each of the investors.

We can then determine the worst-case drift process $(\mu^*_t)_{t\in[0,T]}$ and the optimal strategy $(\pi^*_t)_{t\in[0,T]}$ that is realized by the investor who solves at each time point the local optimization problem
\[ \adjustlimits \sup_{\pi^{(t)}\in\calA_h(t,x)} \inf_{\mu^{(t)}\in \calK^{(t)}} \E_{\mu^{(t)}}\Bigl[U_\gamma\bigl(X^{t,x,\pi^{(t)}}_T\bigr)\Bigr], \]
where the information at $t$ is now given by $\calG_t = \calF^H_t$. 
Recall that $(\mu^*_t)_{t\in[0,T]}$ and $(\pi^*_t)_{t\in[0,T]}$ are calculated from the solutions of the local optimization problems via
\[ \pi^*_t=\pi^{(t),*}_t, \quad \mu^*_t=\mu^{(t),*}_t \]
for all $t\in[0,T]$.
The value of each investor's worst-case optimization is then equal to (writing now  $X_T^\pi$ for $X_T^{0,x_0,\pi}$)
\begin{equation}\label{eq:worst-case_expected_utility}
	\E_{\mu^*}\bigl[U_\gamma(X^{\pi^*}_T)\bigr].
\end{equation}
The quantity in~\eqref{eq:worst-case_expected_utility} is the worst-case expected utility from the $H$-investor's point of view when using the robust strategy $\pi^{*}$. For comparison, we also compute
\[ \E_{\mu^*}\bigl[U_\gamma(X^{\hat{\pi}}_T)\bigr], \quad \E_{\nu}\bigl[U_\gamma(X^{\pi^*}_T)\bigr] \quad\text{and}\quad \E_{\nu}\bigl[U_\gamma(X^{\hat{\pi}}_T)\bigr], \]
where $\nu=\muhat{H}{}$ is the conditional mean of the $H$-investor's filter and $\hat{\pi}$ is the corresponding optimal strategy given that the drift equals $\muhat{H}{}$, i.e.\ 
\[ \hat{\pi}_t=\frac{1}{1-\gamma}A\muhat{H}{t}+hc. \]

We repeat this simulation $10\,000$ times where in each iteration a new drift process, a new return process and new expert opinions are simulated based on the parameters given above. Table~\ref{tab:comparison_of_utility_for_different_investors} gives the sample mean of the various expected utilities over all simulations and in brackets the corresponding sample standard deviation.

\begin{table}[ht]
	\centering
	\begin{tabular}{cc|r|r|r|r}
		$H$	&$n$	&\multicolumn{1}{c}{$\E_{\mu^*}\bigl[U_\gamma(X^{\pi^*}_T)\bigr]$}	&\multicolumn{1}{|c}{$\E_{\mu^*}\bigl[U_\gamma(X^{\hat{\pi}}_T)\bigr]$}	&\multicolumn{1}{|c}{$\E_{\nu}\bigl[U_\gamma(X^{\pi^*}_T)\bigr]$}	&\multicolumn{1}{|c}{$\E_{\nu}\bigl[U_\gamma(X^{\hat{\pi}}_T)\bigr]$}\bstrut\\\hline
		$N$	&	&1.6179 \textit{(0.0000)}	&1.5996 \textit{(0.0000)}	&2.0196 \textit{\phantom{00}(0.0000)}	&2.0426 \textit{\phantom{00\,00}(0.0000)}\tstrut\\
		$R$	&	&1.7086 \textit{(0.1057)}	&0.7754 \textit{(0.3737)}	&2.2362 \textit{\phantom{00}(2.4692)}	&25.9029 \textit{\phantom{00\,}(732.4104)}\\
		$E$	&$10$	&1.7055 \textit{(0.1117)}	&0.8170 \textit{(0.3870)}	&2.2393 \textit{\phantom{00}(3.4208)}	&21.1610 \textit{\phantom{00\,}(530.6829)}\\
		$C$	&$10$	&1.7854 \textit{(0.4027)}	&0.6891 \textit{(0.3752)}	&4.5313 \textit{(134.5858)}	&264.0838 \textit{(19\,288.2826)}\\
	\end{tabular}
	\caption{Comparison of utility for different investors.}\label{tab:comparison_of_utility_for_different_investors}
\end{table}

\paragraph{Observations.}
When comparing the worst-case expected utility $\E_{\mu^*}[U_\gamma(X^{\pi^*}_T)]$ among the investors we see that the information setting $H=N$, which corresponds to only knowing the model parameters, gives the lowest value. The observation of returns or of $n=10$ expert opinions increases this value. The combination of return observation and discrete-time expert opinions yields a considerably larger worst-case expected utility.

In the next column, $\E_{\mu^*}[U_\gamma(X^{\hat{\pi}}_T)]$ measures the expected utility when using strategy $\hat{\pi}$, given that the true drift is actually the worst-case drift $\mu^*$. The values are in any case smaller than the corresponding expected utility when using the robust strategy $\pi^*$. What is striking is that the information setting $H=N$, i.e.\ only knowledge of the model parameters, gives the best expected utility here. Adding more information, from return observations or expert opinions, and using the optimal strategy based on the filter leads to a smaller worst-case expected utility. This shows that for the worst-case optimization problem it is dangerous for investors to rely on their estimates of the drift, i.e.\ the conditional mean of the filter, only. They need to robustify their strategy by taking into account model uncertainty to be able to profit from any additional information. This effect can be linked to overconfidence of experts as studied empirically by Heath and Tversky~\cite{heath_tversky_1991}, seeing that more knowledge about the drift process leads to a worse expected utility in the non-robust case due to taking more risky strategies.

The last two columns show the expected utility when using strategy $\pi^*$, respectively $\hat{\pi}$, given that the true drift was actually the conditional mean $\nu=\muhat{H}{}$. Of course, when compared to the expected utility given the worst-case drift $\mu^*$, the expected utility given $\nu$ is much higher. Not surprisingly, the performance of $\hat{\pi}$ given drift $\nu$ is on average extremely good. However, we also notice the very large sample standard deviation. In comparison to that, we see that the robust strategies $\pi^*$ perform reasonably well given drift $\nu$, even though they are tailored for the worst-case drift in the respective uncertainty set. At the same time, the sample standard deviation is much smaller than for strategy $\hat{\pi}$.

\paragraph{Conclusions.}
In conclusion, we see that a surplus of information, either from return observations or expert opinions, results in better strategies in general. However, investors do need to account for model uncertainty by choosing a robustified strategy $\pi^*$ instead of relying on the respective filter only. The naive strategy $\hat{\pi}$ performs extremely well if the true drift coincides with the conditional mean $\muhat{H}{}$, but it is much more vulnerable to model misspecifications than the robust strategy~$\pi^*$.


\end{document}